\documentclass[12pt]{article}

\usepackage{graphicx}
\usepackage{amssymb}
\usepackage{fancyhdr}
\usepackage{xcolor}
\usepackage{url}
\usepackage{mathtools}
\usepackage{amsmath}
\usepackage[utf8]{inputenc}
\usepackage[english]{babel}
\usepackage{breakcites}
\usepackage{tabu}
\usepackage{mathrsfs}
\usepackage{dsfont}
\usepackage{booktabs}
\usepackage{amsthm}
\usepackage{hyperref}
\usepackage{rotating}
\usepackage{graphicx}
\usepackage{lscape}
\usepackage{enumitem}
\usepackage{bbm}
\usepackage[]{interval}

\intervalconfig{
  soft open fences,
  separator symbol =;,
}
\definecolor{cornellred}{rgb}{0.7, 0.11, 0.11}
\hypersetup{colorlinks,linkcolor={red},citecolor={cornellred},urlcolor={red}}

\newtheoremstyle{break}
  {\topsep}{\topsep}%
  {\itshape}{}%
  {\bfseries}{}%
  {\newline}{}%
  
\theoremstyle{break}

\newtheorem{theorem}{Theorem}
\newtheorem{lemma}{Lemma}

\newtheorem{assumption}{A}

\theoremstyle{definition}

\DeclareMathOperator{\spargel}{sp}

\DeclareMathOperator{\diag}{diag}

\DeclareMathOperator{\E}{\mathbb E}
\DeclareMathOperator{\V}{\mathbb V}

\DeclareMathOperator{\proj}{proj}

\usepackage{geometry}
\newcommand\utimes{\mathbin{\ooalign{$\cup$\cr%
   \hfil\raise0.42ex\hbox{$\scriptscriptstyle\times$}\hfil\cr}}}
\newcommand\bigutimes{\mathop{\ooalign{$\bigcup$\cr%
   \hfil\raise0.36ex\hbox{$\scriptscriptstyle\boldsymbol{\times}$}\hfil\cr}}}

\makeatletter
\renewenvironment{proof}[1][\proofname]{%
  \par\pushQED{\qed}\normalfont%
  \topsep6\p@\@plus6\p@\relax
  \trivlist\item[\hskip\labelsep\bfseries#1\@addpunct{.}]%
  \ignorespaces
}{%
  \popQED\endtrivlist\@endpefalse
}
\makeatother

\newcommand\norm[1]{\left\lVert#1\right\rVert}

\usepackage[round]{natbib}

\makeatletter
\newcommand{\bianca}{\renewcommand\NAT@open{[}\renewcommand\NAT@close{]}}
\makeatother

\providecommand{\keywords}[1]{\textbf{\textit{Index terms---}} #1}

\providecommand{\keywords}[1]{\textbf{\textit{Index terms---}} #1}

\usepackage{stackengine}

\newcommand\ubar[1]{\stackunder[1.2pt]{$#1$}{\rule{.8ex}{.075ex}}}

\begin{document}
%
%
\title{
Actually, There is No Rotational Indeterminacy in the Approximate Factor Model
}
\author{
\textsc{Philipp Gersing}\footnote{Department of Statistics and Operations Research, University of Vienna}}
\maketitle
\begin{abstract}
We show that in the approximate factor model the population normalised principal components converge in mean square (up to sign) under the standard assumptions for $n\to \infty$. Consequently, we have a generic interpretation of what the principal components estimator is actually identifying and existing results on factor identification are reinforced and refined. Based on this result, we provide a new asymptotic theory for the approximate factor model entirely without rotation matrices. We show that the factors space is consistently estimated with finite $T$ for $n\to \infty$ while consistency of the factors a.k.a the $L^2$ limit of the normalised principal components requires that both $(n, T)\to \infty$. 
%
%
%
%
%
%
%
%
%
\end{abstract}
%
%
%
%
%
%
\section{Introduction}
Let $(y_{it}: i \in \mathbb N, t \in \mathbb Z) \equiv (y_{it})$ be a zero mean stationary stochastic process indexed in time $t$ and cross-section $i$. The approximate factor model introduced by \cite{chamberlain1983arbitrage, chamberlain1983funds, stock2002forecasting, stock2002macroeconomic, bai2002determining} relies on a decomposition of the form
\begin{align}
    y_{it} &= C_{it} + e_{it} = \Lambda_i F_t + e_{it}, \label{eq: r-SFM} \\
    y_t^n &= C_t^n + e_t^n = \Lambda^n F_t + e_t^n \quad \mbox{(vector representation)} \label{eq: r-SFM vec rep}
\end{align}
where $F_t$ are factors of dimension $r$ (usually small), the $\Lambda_i$'s are pervasive loadings with $\Lambda^n = (\Lambda_1', ..., \Lambda_n')'$, $C_{it}$ is the ``common component'' and $e_{it}$ is the contemporaneously weakly correlated ``idiosyncratic component''. We consider a setup where the loadings $\Lambda_i$ are ``deterministic parameters''. 

Clearly, the ``true factors'' are latent and can only be identified up to a non-singular transformation. A classical result \citep[see][]{stock2002forecasting, bai2002determining, bai2003inferential, barigozzi2022estimation}, what we may call ``convergence to the factor space'' is 
\begin{align}
    \norm{\hat F_t^n - \hat H_n F_t} \overset{P}{\to} 0 \quad \mbox{for} \ n, T \to \infty, \label{eq: convergence with hat H}
\end{align}
while $\hat F_t^n$ denotes factors estimated by the normalised (by square root for the sample eigenvalues) principal components (obtained from the sample variance matrix) $\hat H_n$ is a non-singular transformation depending on $n$. Alternatively, there is also a result where $H_n$ is deterministic \citep[see][]{barigozzi2022estimation}.
%

Next, \cite{bai2013principal} assume the identifying restrictions 
\begin{itemize}
    \item[(i)] $(\Lambda^n)' \Lambda^n \in \mathcal D(r)$ is diagonal for all $n \in \mathbb N$ 
    \item[(ii)] $T^{-1}\sum_{t = 1}^T F_t F_t' = I_r$
\end{itemize}
To maintain the property $(\Lambda^n)' \Lambda^n$ is diagonal \textit{for all} $n\in \mathbb N$ the loadings $\Lambda_i$ associated with a fixed $i$ need to change with every $n$ in general: For instance, let $\lambda_{ij}$ be the $i,j$ entry of the loadings matrix and $\lambda_j^n$ be the $j$-th column, then in the following case $(\Lambda^2)'\Lambda^2$ is diagonal, but for
\begin{align*}
\Lambda^3 &= \begin{pmatrix}
        1 & 0 \\
        0 & 1 \\
\lambda_{31} & \lambda_{32}
    \end{pmatrix}
    \quad \mbox{then} \ (\Lambda^3)'\Lambda^3 \ \mbox{is diagonal only if} \ 
    (\lambda_1^3)' \lambda_2^3 = 0 \\[0.8em]
    &\mbox{which requires} \ \lambda_{31} \lambda_{32} = 0.
\end{align*}
So under the restrictions $(i), (ii)$ above, without allowing $\Lambda_i$ to depend on $n$, it would be too restrictive to demand orthogonality in the loadings columns for every $n\in \mathbb N$. The same holds for condition (ii) - every trajectory is associated with a another factor representation. As is shown in \cite{bai2013principal, barigozzi2022estimation} the identification restrictions above imply that under standard assumptions  
\begin{align}
    \norm{\hat F_t^n - \hat F_t^{n,C}} \overset{P}{\to} 0 \quad \mbox{for} \ n, T \to \infty, \label{eq: convergence to hat NPC of C}
\end{align}
where $\hat F_t^{n, C}$ are the normalised sample principal components of $C_t^n$ for every $n$. We can be shown to converge to $F_t^{n, C}$ the population principal components in probability for $T\to \infty$. However, nothing has been said about whether $\hat F_t^{n,C}$ or $F_t^{n,C}$ converges or not. It is well known that in general, the principal components of a multivariate vector $y_t^n$ or $C_t^n$ change with $n$, since adding new variables changes the principal directions of the data. Therefore we conclude that (\ref{eq: convergence to hat NPC of C}) establishes \textit{convergence to sequence of factor representations} - not to a specific factor representation. Often we wish to interpret the loadings structurally, e.g. achieve a representation in which $\Lambda^n$ is sparse. In an asymptotic setup which requires $n \to \infty$ for identification, convergence to a sequence of representations is not appropriate for interpretations of the loadings, like e.g. sparsity: After all, as shown above if $\Lambda^n$ is sparse under restrictions (i), (ii), then $\Lambda^{n+1}$ in general is not sparse.

Finally \cite{bai2020simpler} prove for $\hat H_n$ in (\ref{eq: convergence with hat H}) that $\hat H_n \overset{P}{\to} H_0$. This is wonderful as combining this with (\ref{eq: convergence with hat H}) it readily implies that
\begin{align}
    \norm{\hat F_t^n - H_0 F_t} \leq  \underbrace{\norm{\hat F_t^n - \hat H_n F_t}}_{\overset{P}{\to} 0} + \underbrace{\norm{\hat H_n - H_0}}_{\overset{P}{\to} 0}\underbrace{\norm{F_t}}_{\mathcal O_p(1)}\overset{P}{\to} 0 \quad \mbox{for} \ n, T \to \infty, \label{eq: convergence to factors with H_0}
\end{align}
where here the factor representation $F_t$ does depend on $n$ as in (\ref{eq: r-SFM}). Therefore the sample normalised principal components $\hat F_t^n$ have a probability limit for $(n, T)\to \infty$, which is $H_0 F_t$. We may say that (\ref{eq: convergence to factors with H_0}) establishes \textit{convergence to a specific factor representation} in probability.  

The contribution of this paper is to extend/refine this result by proving that under suitable (standard conditions) the normalised population principal components of $C_t^n$ and $y_t^n$, say $F_t^{n, C}$ and $F_t^{n, y}$ converge in mean square for $n\to \infty$ to what we may call \textit{the normalised principal components of the statically common space}, say $F_t^\infty$, i.e. 
\begin{align}
    \norm{F_t^{n, y}- F_t^\infty} = \mathcal O_{ms}(n^{-1/2}) \quad \mbox{and} \quad  \norm{F_t^{n, C}- F_t^\infty} = \mathcal O_{ms}(n^{-1/2}). \label{eq: L2 convergence of NPCs}
\end{align}
- a representation independent of $n$. On the filp side also the associated loadings $\Lambda_i^\infty$ are independent of $n$. 

This result paves the way for a simpler asymptotic theory, presented in the second part of the paper, that operates without ``rotation matrices'' like $\hat H_n$ or $H_n$ and establishes convergence to factors rather than a sequence of factors. The proofs are approached like in \cite{forni2004generalized} using perturbation theory from \cite{wilkinson1965algebraic}. Consequently, we have a generic interpretation for what the principal components estimator is actually targeting and structural interpretations on the estimated loadings, like e.g. sparsity are justified. Furthermore it also provides meaningful interpretations for regression outputs in factor-augmented regressions \citep[see][]{bernanke2005measuring, bai2006confidence}, as we have a meaningful interpretation of each of the factors. In particular we show that
\begin{align*}
    \norm{\hat F_t^n - F_t^\infty} \leq \underbrace{\norm{\hat F_t^n - F_t^{n, y}}}_{\mathcal O_{ms}(T^{-1/2})} + \underbrace{\norm{F_t^{n, y} - F_t^\infty}}_{\mathcal O_{ms}(n^{-1/2})} = \mathcal O_{ms}\left(\max(T^{-1/2}, n^{-1/2})\right), 
\end{align*}
which clarifies the different roles of $T$ and $n$ in the identification. While $T\to\infty$ is needed to estimate the normalised principal components, $n\to\infty$ is needed to average out the idiosyncratic terms.

Finally, we also clarify why consistency of the factor space \textit{only requires} $n\to \infty$, i.e. we show
\begin{align*}
    \norm{\hat F_t^n - \hat H_n F_t} = \mathcal O_{ms}(n^{-1/2}) \quad \mbox{for finite} \ T < \infty.
\end{align*}
Convergence of the factor space under finite $T <\infty$ has also been studied and proved in \cite{zaffaroni2019factor, fortin2023eigenvalue, onatski2023comment} under slightly different assumptions and the earlier in \cite{connor1986performance, bai2003inferential}. The proof presented in this paper reveals how and why this result seamlessly fits with the approximate factor model framework: To identify the factor space, we only need to average out the idiosyncratic term by cross-sectional aggregation \citep[see][]{gersing2023reconciling, gersing2023weak}. If we have $n = \infty$, the factor space is obtained whenever, we can find cross-sectional averages that result in $r$ linearly independent factors. This is satisfied even for very imprecisely estimated eigenvectors in the case of finite $T$.   
%
%

    %
    %
%
%
%
%
\section{General Setup}
We suppose that the observed process $(y_{it})$ has the representation as in (\ref{eq: r-SFM}), (\ref{eq: r-SFM vec rep}). Let $A$ be a symmetric $n\times n$ matrix. Denote by $\mu_j(A)$ the $j$-th largest eigenvalue of a matrix $A$ and set $M(A) \equiv \diag(\mu_1(A), ..., \mu_r(A))$. Denote by $P(A)$ the $r\times n$ matrix consisting of the first orthonormal row eigenvectors (corresponding to the $r$ largest eigenvalues of $A$) and by $P^i(A)$ the $1\times r$ vector consisting of the entries of the $i$-th row of $P(A)'$. Let $\E y_t^n (y_t^n)' \equiv \Gamma_y^n$ with $y_t^n = (y_{1t}, ..., y_{nt})'$.
The normalised principal components of $y_t^n$ are 
\begin{align*}
    F_t^{n, y} &\equiv M^{-1/2}(\Gamma_y^n) P(\Gamma_y^n) y_t^n . \\
    F_t^{n, C} &\equiv M^{-1/2}(\Gamma_C^n) P(\Gamma_C^n) C_t^n 
\end{align*}
For simplicity we suppose that the loadings are such that $\E F_t F_t' = I_r$ which can always be achieved by a constant (independent of $n$) rotation matrix applied to the factors. Furthermore the following assumptions are made: 
\begin{assumption}[r-Static Factor Structure]\label{A: r-SFS struct}
There exists a natural number $r < \infty$, such that
\begin{itemize}
    \item[(i)] $\sup_n \mu_r(\Gamma_C^n) = \infty$ (pervasive loadings)
    \item[(ii)] $\sup_n \mu_1(\Gamma^n_e) < \infty$. 
\end{itemize}
\end{assumption}
Note that part $(ii)$ is also implied by assumptions made for the idiosyncratic component as in \cite{bai2002determining}. This is shown in \cite{barigozzi2022estimation}. Henceforth let us commence from a representation in terms of ``true factors'' which is such that $\E F_t F_t' = I_r$. The existence of such a representation is implied by A\ref{A: r-SFS struct} \citep[see][which uses the same proofs  as \cite{forni2001generalized} operating with variance matrices insted of spectral densities]{gersing2023reconciling}. It also follows from the classical assumptions of \cite{bai2002determining, bai2003inferential}. 
\begin{assumption}[Asymptotic Properties of the Loadings]\label{A: divergence rates eval}\ \\[-2.5em]
    \begin{itemize}
    \item[(i)] Normalised Representation: Without loss of generality we assume that $\E F_t F_t' = I_r$
    \item[(ii)] Convergence of the ``loadings variance'': Set $\frac{(\Lambda^n)'\Lambda^n}{n} = \frac{1}{n} \sum_{i = 1}^n \Lambda_i' \Lambda_i \equiv \Gamma_\Lambda^n$ and suppose that $\norm{\Gamma_\Lambda^n- \Gamma_\Lambda} = \mathcal O(n^{-1/2})$ for some $\Gamma_\Lambda > 0$.
    \item[(iii)] The eigenvalues of $\Gamma_\Lambda$ are distinct and contained in the diagonal matrix $D_\Lambda$ sorted from the largest to the smallest.
    \item[(iv)] Loadings and idiosyncratic variances are globally bounded, i.e. $\norm{\Lambda_i} < B_\Lambda < \infty$ and $\E e_{it}^2 < B_e < \infty$ for all $i\in \mathbb N$.  
    \end{itemize}
\end{assumption}
%
%
%
%
%
%
%
As we will see this is equivalent to supposing that $n^{-1}M(\Gamma_y^n) \to D_\Gamma = D_\Lambda$ are distinct. We can think of A\ref{A: divergence rates eval}(i) in the sense of a law of large numbers for the loadings, e.g. if the loadings/cross-sectional units are sampled in a first step from an IID distribution. The eigenvalues divided by the rate converge to diagonal matrix $D_\Gamma = D_\Lambda$, are asymptotically well separated. Convergence \textit{and} separation are both crucial for $L^2$ convergence of the normalised principal components. In the literature \citep[][]{bai2002determining, barigozzi2022estimation}, it is usually assumed that $n^{-1} (\Lambda^n)' \Lambda^n \to \Gamma_\Lambda > 0$ and $\Gamma_F > 0$ which implies the convergence of $n^{-1}M(\Gamma_C^n)$ as we will see below. Asymptotically separated eigenvalues assumed e.g. in \cite{forni2004generalized} (Assumption R, but without convergence) in \cite{bai2003inferential} (Assumption G) and in \cite{bai2013principal}. 
%
%
%
%
%
\section{The Convergence Result}
Clearly the normalised principal components of $F_t^{n, y}$ are determined only up to sign. To resolve this indeterminancy, lets assume (without loss of generality) that the first $r$ rows of $P'(\Gamma_y^n) M^{1/2}(\Gamma_y^n)$, call it $\Lambda_r(\Gamma_y^n)$ have full rank (from a certain $n$ onwards) and fix the diagonal elements of $\Lambda_r(\Gamma_y^n)$ to be always positive. This fixes the sign of the eigenvectors and the normlised principal components. 
\begin{lemma}\label{lem: evals-vecs Gammay GammaC}
    Under Assumption A\ref{A: r-SFS struct}, we have
    \begin{itemize}
        \item[(i)] $\norm{p_j(\Gamma_y^n) - p_j(\Gamma_C^n)} = \mathcal O(n^{-1/2})$, for $j = 1, ..., r$.
        \item[(ii)] $\mu_j(\Gamma_y^n) - \mu_j(\Gamma_C^n) = \mathcal O(1)$ for $1\leq j \leq r$. 
        \item[(iii)] $\frac{M(\Gamma_y^n)}{n}\to D_\Gamma, \frac{M(\Gamma_C^n)}{n}\to D_\Gamma$ for $n\to \infty$ and $|\frac{M(\Gamma_y^n)}{n} - \frac{M(\Gamma_C^n)}{n}| = \mathcal O(n^{-1})$
    \end{itemize}
\end{lemma}
\begin{proof}
    (i) Consider for any fixed $n$ a symmetric positive semi-definite matrix $\Gamma^n$. Suppose we perturb $\Gamma^n$ such that $\hat \Gamma^n = \Gamma^n + \varepsilon_n B^n$, where the entries of $B^n$ are in modulus smaller than one. Let $p_l \equiv p_l(\Gamma^n), \mu_j \equiv \mu_j(\Gamma^n)$ for $l = 1,..., n$, suppressing dependence on $n$ in the notation. From \cite{wilkinson1965algebraic}, section 2,(10.2) we obtain the expansion that if $\varepsilon_n$ is sufficiently small, that 
    \begin{align*}
        p_j(\hat \Gamma^n) - p_j(\Gamma^n) &= \sum_{l \neq j, l = 1}^n \frac{p_l\left(\hat \Gamma^n - \Gamma^n\right)p_j'}{\mu_j - \mu_l} p_l \ \varepsilon_n + \cdots \nonumber \\
        \mbox{therefore} \quad  \norm{p_j(\hat \Gamma^n) - p_j(\Gamma^n)}^2 &= \sum_{l \neq j, l = 1}^n \left(\frac{p_l\left(\hat \Gamma^n - \Gamma^n\right)p_j'}{\mu_j - \mu_l} \right)^2 \varepsilon_n^2 + \cdots, 
    \end{align*}
where the higher order terms are smaller than the first order terms and can be neglected and the second equation follows by orthogonality of the eigenvectors and while the higher order terms are negligent.

Apply this to $\hat \Gamma^n = \Gamma_y^n$ and $\Gamma^n = \Gamma_C^n$, so $p_l\equiv p_l(\Gamma_y^n), \mu_j \equiv \mu_j(\Gamma_y^n)$, we obtain 
\begin{align*}
    \norm{p_j(\Gamma_y^n) - p_j(\Gamma_C^n)}^2 &= \sum_{l \neq j, l = 1}^n \left(\frac{p_l \Gamma_e^n p_j'}{\mu_j - \mu_l} \right)^2 \varepsilon_n^2 + \cdots \\[1em] 
    \sum_{l \neq j, l = 1}^n \left(\frac{p_l \Gamma_e^n p_j'}{\mu_j - \mu_l} \right)^2 & \leq  \left(\sup_n \mu_1(\Gamma_e^n)\right)^2 \sum_{l \neq j, l = 1}^n  \sup_{l\leq n} \left(\frac{1}{\mu_j - \mu_l}\right)^2   \\
    & \leq n^{-2} \ n B_e^2 \ \sup_{l\leq n} \left(\frac{1}{n^{-1} (\mu_j - \mu_l)}\right)^2  \quad \mbox{by A\ref{A: r-SFS struct}}  \\
    & \leq n^{-2} \ n  B_e^2 \ \max_{l\neq j, 1\leq l, j \leq r}\left(\frac{1}{c_l^-  - c_j^+}\right)^2 = \mathcal O(n^{-1})   \quad \mbox{by A\ref{A: divergence rates eval},(i)},
\end{align*}
where the last inequality holds for $n$ large enough. 

(ii) For the eigenvalues, we use the expansion from \cite{wilkinson1965algebraic}, section 2, (5.5) with coefficients $k_1, k_2, ...$ and $\varepsilon = \mathcal O(1)$ such that 
\begin{align*}
    &\mu_1(\Gamma_y^n) - \mu_1(\Gamma_C^n) = k_1 \varepsilon + k_2 \varepsilon^2 + \cdots\\[0.5em]
    &k_1 = p_1(\Gamma_C^n) \left(\Gamma_y^n - \Gamma_C^n\right) p_1' (\Gamma_C^n)  = p_1(\Gamma_C^n) \left(\Gamma_e^n\right) p_1' (\Gamma_C^n) \leq \mu_1(\Gamma_e^n) = \mathcal O(1)
\end{align*}
where the higher order terms are negligent.

(iii) It readily follows with $\mu_j(\Gamma_C^n) = \mu_j\left(\Lambda^n (\Lambda^n)'\right) = \mathcal O(n)$ and under A\ref{A: divergence rates eval}, that $M(\Gamma_y^n) / n \to D_\Gamma$. 
\end{proof}

The normalised static principal components (NSPC) of any random vector of dimension $n$, supposing that the eigenvalues are different, are unique up to a change in sign. Suppose now, we fix the sign. Clearly, the NSPCs do in general not converge if $n$ is successively increased, since a new cross-sectional unit may always alter the direction of the principal components. However, in the case of the approximate factor model and under the condition of asymptotic stability of the first $r$-eigenvalues (which implies that they are well separated) convergence is ensured.

Consequently, under these assumptions the limit of the NSPC have a meaningful interpretation as ``normalised principal components of the static aggregation space''. The limit of the first NSPC is the ``strongest direction'' in the static aggregation space, the limit of the second NSPC is the ``strongest direction orthogonal to the first'' and so on... 
\begin{theorem}
Under Assumptions A\ref{A: r-SFS struct} and A\ref{A: divergence rates eval}, there exists a mean square limit of the normalised principal components, say $F_t^\infty$.
\begin{itemize}
    \item[(i)] $\norm{M^{-1/2}(\Gamma_y^n) P(\Gamma_y^n)y_t^n - F_t^\infty} = \mathcal O_{ms}(n^{-1/2})$
    \item[(ii)] $\norm{M^{-1/2}(\Gamma_C^n) P(\Gamma_C^n) C_t^n - F_t^\infty} = \mathcal O_{ms}(n^{-1/2})$.
\end{itemize}
\end{theorem}
We may interpret $F_t^\infty$ as the normalised principal components of the static aggregation space. 
\begin{proof}
We want to show that 
\begin{align*}
   F_t^{n, y} \equiv M^{-1/2}(\Gamma_y^n) P(\Gamma_y^n) y_t^n = M^{-1/2}(\Gamma_y^n)  P(\Gamma_y^n) \Lambda^n F_t + M^{-1/2}(\Gamma_y^n)  P(\Gamma_y^n) e_t^n 
\end{align*}
converges. Again, set $\mu_j \equiv \mu_j(\Gamma_y^n), p_j^n \equiv p_j(\Gamma_y^n)$. The second term on the RHS converges to zero as
\begin{align*}
\E (\mu_j^{-1/2} p_j^n e_t^n)^2 \leq \mu_j^{-1}\lambda_1(\Gamma_e^n), \quad \mbox{so} \ \mu_j^{-1/2} p_j^n e_t^n = \mathcal O_{ms}(n^{-1/2}).
\end{align*}
We are left with showing convergence of the first term. Consider the eigen-decompositions:
\begin{align*}
    P_\Lambda^n \Gamma_\Lambda^n (P_\Lambda^n)' = D_\Lambda^n 
    \quad &\mbox{and} \quad
    P_\Lambda \Gamma_\Lambda P_\Lambda' = D_\Lambda \\[0.5em]
    \quad \mbox{and set} \quad  \tilde P_n \equiv D_\Lambda^{-1/2} P_\Lambda \frac{(\Lambda^n)'}{\sqrt{n}}
    \quad &\mbox{and} \quad
    P_n \equiv (D_\Lambda^n)^{-1/2} P_\Lambda^n \frac{(\Lambda^n)'}{\sqrt{n}}.
\end{align*}
It follows that $P_n$ are orthonormal row-eigenvectors of $\Gamma_C^n$ since 
\begin{align*}
    P_n P_n' &= I_r\\
    \mbox{and} \quad P_n \Gamma_C^n &= (D_\Lambda^n)^{-1/2} P_\Lambda^n \frac{(\Lambda^n)'}{\sqrt{n}}\frac{\Lambda^n (\Lambda^n)'}{n} n \\
    &= (D_\Lambda^n)^{-1/2} P_\Lambda^n \Gamma_\Lambda^n \frac{(\Lambda^n)'}{\sqrt{n}} n  = (D_\Lambda^n)^{-1/2} D_\Lambda^n P_\Lambda^n \frac{(\Lambda^n)'}{\sqrt{n}}n = D_\Lambda^n n P_n = M(\Gamma_C^n) P_n.
\end{align*}
In order two show that $F_t^{n, y}$ converges, it is enough to show that 
\begin{align*}
    P(\Gamma_y^n) \frac{\Lambda^n}{\sqrt{n}} F_t 
\end{align*}
converges, which is the case whenever $P(\Gamma_y^n) \frac{\Lambda^n}{\sqrt{n}}$ converges to a constant finite $r\times r$ matrix. 
So 
\begin{align*}
   P(\Gamma_y^n) \frac{\Lambda^n}{\sqrt{n}} - D_\Lambda^{-1/2} P_\Lambda \Gamma_\Lambda = \underbrace{\tilde P_n \frac{\Lambda^n}{\sqrt{n}} - D_\Lambda^{-1/2} P_\Lambda \Gamma_\Lambda}_{(I)} + \underbrace{\left[ P(\Gamma_y^n) - P_n\right] \frac{\Lambda^n}{\sqrt{n}}}_{(II)} + \underbrace{\left[P_n - \tilde P_n \right]\frac{\Lambda^n}{\sqrt{n}}}_{(III)}. 
\end{align*}
where $\norm{\frac{\Lambda^n}{\sqrt{n}}} = \mathcal O(1)$ so $\norm{(II)} = \mathcal O(n^{-1/2})$ by Lemma \ref{lem: evals-vecs Gammay GammaC}. For $(III)$ term note that 
\begin{align*}
    P_n - \tilde P_n = \left[\left(D_\Lambda^n\right)^{-1/2} P_\Lambda^n - 
 D_\Lambda^{-1/2} P_\Lambda  \right]\frac{(\Lambda^n)'}{\sqrt{n}}, 
\end{align*}
while again by \cite{wilkinson1965algebraic} we have with $B_n = n^{1/2}c_n^{-1} (\Gamma_\Lambda^n - \Gamma_\Lambda) = \mathcal O(1)$ by A\ref{A: divergence rates eval}(i), where $c_n$ is the modulus of the maximum entry of $n^{1/2}(\Gamma_\Lambda^n - \Gamma_\Lambda)$ plus some small constant (rendering the entries of $B_n$ smalle than one in modulus) and $\varepsilon_n \equiv n^{-1/2}c_n$, 
\begin{align}
     \norm{p_j(\Gamma_\Lambda^n) - p_j(\Gamma_\Lambda)}^2 &= \sum_{l \neq j, l = 1}^r \left(\frac{p_l\left(\Gamma_\Lambda\right)B_n p_j'(\Gamma_\Lambda)}{\mu_j(\Gamma_\Lambda) - \mu_l(\Gamma_\Lambda)} \right)^2 \varepsilon_n^2 + \cdots = \mathcal O(n^{-1}), \label{eq: evec approx Gamma Lambda}
\end{align}
noting that by A\ref{A: divergence rates eval}(ii) the eigenvalues of $\Gamma_\Lambda$ are distinct, so the denominator on the RHS in equation (\ref{eq: evec approx Gamma Lambda}) does not vanish, consequently $\norm{P_\Lambda^n - P_\Lambda} = \mathcal O(n^{-1/2})$.

Next, for the eigenvalues, we have
\begin{align*}
    &\mu_j(\Gamma_\Lambda^n) - \mu_j(\Gamma_\Lambda) = k_1 \varepsilon_n + k_2 \varepsilon_n^2 + \cdots\\[0.5em]
 \mbox{with} \quad  &k_1 = p_j(\Gamma_\Lambda) B_n p_j' (\Gamma_\Lambda) = \mathcal O(1) \quad j = 1, ..., r
\end{align*}
So $|\mu_j(\Gamma_\Lambda^n) -\mu_j(\Gamma_\Lambda)| = \mathcal O(n^{-1/2})$ and therefore 
\begin{align*}
|\mu_j(\Gamma_\Lambda^n)^{-1} - \mu_j(\Gamma_\Lambda)^{-1}| \leq |\mu_j\left(\Gamma_\Lambda^n\right)^{-1}| |\mu_j\left(\Gamma_\Lambda^n\right) - \mu_j\left(\Gamma_\Lambda\right)||\mu_j(\Gamma_\Lambda)^{-1}| = \mathcal O(n^{-1/2})
\end{align*}
which implies by $(a^2-b^2) = (a-b)/ (a+b)$ for $a, b \in \mathbb R, a+b\neq 0$ that 
\begin{align*}
\norm{(D_\Lambda^n)^{-1/2} - D_\Lambda^{-1/2}} = \mathcal O(n^{-1/2}).
\end{align*}
As a result $(III) = \mathcal O(n^{-1/2})$ since $\norm{\tilde P_n - P_n} = \mathcal O(n^{-1/2})$.
Finally we get
\begin{align*}
    \norm{(I)} = \norm{D_\Lambda^{-1/2} P_\Lambda(\Gamma_\Lambda^n - \Gamma_\Lambda)} = \mathcal O(n^{-1/2}). 
\end{align*}
Now obviously we have 
\begin{align*}
    & P(\Gamma_y^n)\frac{\Lambda^n}{\sqrt{n}} F_t \to D_\Lambda^{-1/2} P_\Lambda \Gamma_\Lambda F_t = D_\Lambda^{1/2} P_\Lambda F_t \\[0.8em]
    \mbox{and therefore} \quad & \norm{\left(\frac{M(\Gamma_y^n)}{n}\right)^{-1} P(\Gamma_y^n) \frac{\Lambda^n}{\sqrt{n}} F_t - P_\Lambda F_t} = \mathcal O_{ms}(n^{-1/2}), 
\end{align*}
so $P_\Lambda F_t = F_t^\infty$. The second statement follows by analogous arguments.  
\end{proof}
\begin{theorem}
Under Assumptions A\ref{A: r-SFS struct} and A\ref{A: divergence rates eval}, there exists a limit of the loadings $\Lambda_i(\Gamma_y^n) \equiv P^i (\Gamma_y^n) M^{1/2}(\Gamma_y^n)$, say $\Lambda_i^\infty$ with
\begin{itemize}
    \item[(i)] $\norm{\Lambda_i(\Gamma_y^n) - \Lambda_i^\infty} = \mathcal O(n^{-1/2})$
    \item[(ii)] $\norm{P^i (\Gamma_C^n) M^{1/2}(\Gamma_C^n) - \Lambda_i^\infty} = \mathcal O(n^{-1/2})$
    \item[(iii)] $\norm{\frac{1}{n}\sum_{i = 1}^n (\Lambda_i^\infty)' \Lambda_i^\infty - D_\Lambda} = \mathcal O(n^{-1/2})$
\end{itemize}
\end{theorem}
\begin{proof}
 One way of proving the convergence of the loadings would be by looking directly at the limit of $\sqrt{n}P^i(\Gamma_y^n) \left(\frac{M(\Gamma_y^n)}{n}\right)^{1/2}$, using the approximation of eigenvectors and eigenvalues from above. Alternatively, note that
 \begin{align*}
    C_{it} = \proj\left(y_{it}\mid \spargel(F_t^\infty)\right) &= \Lambda_i^\infty F_t^\infty\\
    \mbox{and} \quad \proj\left(y_{it}\mid \spargel(F_t^{n, y})\right) &= \Lambda_i(\Gamma_y^n) F_t^{n, y}.
 \end{align*}
since by construction $\V F_t^{n, y} = I_r$, we know that $\E F_t^{n,y}y_{it} = \Lambda_i(\Gamma_y^n)'$. Now check the limit of $\Lambda_i(\Gamma_y^n)'$ by
\begin{align*}
    \E F_t^{n, y} y_{it} - (\Lambda_i^\infty)'  &= \E \left[F_t^{n, y} (\Lambda_i^\infty F_t^\infty + e_{it})'\right] - (\Lambda_i^\infty)'\\
    &=\E \left[(F_t^{n, y} - F_t^\infty + F_t^\infty) e_{it}\right] + \left\{\E\left[(F_t^{n, y} - F_t^\infty + F_t^\infty) (F_t^\infty)'\right] - I_r\right\} (\Lambda_i^\infty)' \\
    &= \E\left[ (F_t^\infty - F_t^{n, y}) e_{it}\right] + \E \left[(F_t^{n, y} - F_t^\infty) (F_t^\infty)'\right] (\Lambda_i^\infty)' \\[0.8em]
    \mbox{so} \quad \norm{\Lambda_i(\Gamma_y^n) - \Lambda_i^\infty} &\leq  \norm{F_t^{n, y} - F_t^\infty}_{L^2} \norm{e_{it}}_{L^2} + \norm{F_t^{n, y} - F_t^\infty}_{L^2} \norm{F_t^\infty}_{L^2} \norm{\Lambda_i^\infty} \\
    &=\mathcal O(n^{-1/2}) \mathcal O(1) + \mathcal O(n^{-1/2})\mathcal O(1) \mathcal O(1) = \mathcal O(n^{-1/2}),  
\end{align*}
where we used that $\V F_t^\infty = I_r$ and $\E F_t^\infty e_{it} = 0$ in the third line and Cauchy-Schwarz' inequality in the fourth line. The second statement is proved with analogous arguments.

It immediately follows that $\frac{1}{n}\sum_{i = 1}^n (\Lambda_i^\infty)' \Lambda_i^\infty$ converges to the diagonal matrix $D_\Lambda$ with rate $n^{-1/2}$.
\end{proof}
\section{Asymptotic Theory without Rotation Matrices}
\begin{assumption}[Estimation of Variance Matrices]\label{A: Var est}\ \\[-2.5em]
    \begin{itemize}
         \item[(ii)] $\norm{\hat \Gamma_F - \Gamma_F} = \mathcal O_{ms}(T^{-1/2})$, where $\hat \Gamma_F = T^{-1}\sum_{t=1}^T F_t F_t'$.
          \item[(iv)] $\norm{\hat \Gamma_{Fe}^n p - \Gamma_{Fe}^n p} = \mathcal O_{ms}(T^{-1/2})$ for all $n \in \mathbb N$ with for all $p$ with $\norm{p} = 1$, where $\Gamma_{Fe}^n \equiv  \E F_t (e_t^n)'$ and $\hat \Gamma_{Fe}^n \equiv T^{-1}\sum_{t = 1}^T F_t (e_t^n)'$. 
        \item[(v)] $p(\hat \Gamma_e^n - \Gamma_e^n)p' = \mathcal O_{ms}(T^{-1/2})$ for all $n\in \mathbb N$ with for all $p$ with $\norm{p} = 1$, where $\E e_t^n (e_t^n)' \equiv \Gamma_e^n$ and $\hat \Gamma_e^n = T^{-1}\sum_{t = 1}^T e_t^n (e_t^n)'$
    \end{itemize}
\end{assumption}
%
%
%
%
%
%
%
%
%
Primitive conditions for this assumptions can be given along the lines of \cite{barigozzi2022estimation}.

%
%
\begin{lemma}
    Under Assumptions A\ref{A: r-SFS struct}-A\ref{A: Var est}, we have 
    \begin{itemize}
        \item[(i)] $\norm{\hat p_j^n - p_j^n} = \mathcal O_{ms}(T^{-1/2})$
        \item[(ii)] $\norm{\frac{M(\hat\Gamma_y^n)}{n}  - \frac{M(\Gamma_y^n)}{n}} = \mathcal O_{ms}(T^{-1/2})$
        \item[(iii)] $|\mu_j(\hat \Gamma_y^n)^{-1/2} - \mu_j(\Gamma_y^n)^{-1/2}| = \mathcal O_{ms}(n^{-1/2}T^{-1/2})$.
        \item[(iv)] $|\hat p_{ij}^n - p_{ij}^n| = \mathcal O_{ms}(n^{-1/2}T^{-1/2})$
        \item[(v)] $\norm{\hat{\mathcal K}_j^{F, n} - \mathcal K_j^{F, n}} = \mathcal O_{ms}(T^{-1/2}n^{-1/2})$
        %
        %
    \end{itemize}
    \end{lemma}
\begin{proof}
(i) Set $\hat \Gamma^n = \hat \Gamma_y^n = T^{-1} \sum_{t = 1}^T y_t^n (y_t^n)'$ and $\Gamma^n = \Gamma_y^n$. Set $\hat \Gamma_F = T^{-1} \sum_{t = 1}^T F_t F_t'$, $\hat \Gamma_{eF}^n = T^{-1} \sum_{t = 1}^T e_t^n F_t'$ and $\hat \Gamma_e^n = T^{-1} \sum_{t = 1}^T e_t^n (e_t^n)'$, so
{
\begin{align}
    t_{lj} &= p_l(\Gamma_y^n) \left(\hat \Gamma_y^n  - \Gamma_y^n\right) p_j(\Gamma_y^n)' \nonumber \\[0.5em]
    &=   p_l(\Gamma_y^n) \left( \Lambda^n \hat \Gamma_F \left(\Lambda^n\right)' - \Lambda^n\left(\Lambda^n\right)' +  \Lambda^n \hat \Gamma_{Fe}^n +  \hat \Gamma_{eF}^n (\Lambda^n)' + \hat \Gamma_e^n - \Gamma_e^n \right) p_j(\Gamma_y^n)'  \nonumber \\[0.5em]
    &= p_l(\Gamma_y^n)\Lambda^n \left(\hat \Gamma_F - I_r \right)(\Lambda^n)'p_j(\Gamma_y^n)' + p_l(\Gamma_y^n)\Lambda^n \hat \Gamma_{Fe}^n p_j(\Gamma_y^n)' \nonumber \\[0.5em]
    &  \hspace{7.5cm} + p_l(\Gamma_y^n) \hat \Gamma_{eF}^n (\Lambda^n)'p_j(\Gamma_y^n)' + p_l(\Gamma_y^n)\left(\hat \Gamma_e^n - \Gamma_e^n\right) p_j(\Gamma_y^n)' \nonumber \\[0.5em]
    &= (I)_{lj} + (II)_{lj} + (III)_{lj} + (IV)_{lj}, \ \mbox{say}. \label{eq: terms for evec approx.}
\end{align}
}
Firstly, setting $p_l \equiv p_l(\Gamma_y^n), p_l^C \equiv p_l(\Gamma_C^n)$ for $1\leq l \leq n$, $n \in \mathbb N$, and $\mu_{1, n}^{1/2} \equiv \sqrt{\mu_1(\Gamma_C^n)}$, we have
\begin{align*}
    (I)_{lj} &\leq \norm{p_l \Lambda^n} \norm{p_j \Lambda^n} \mu_1(\hat \Gamma_F - I_r) \\
  \mbox{where} \quad \norm{p_l \Lambda^n} &= \norm{p_l^C\Lambda^n + \underbrace{(p_l - p_l^C)}_{q_l} \Lambda^n} \leq \norm{p_l^C \Lambda^n} + \norm{q_l \Lambda^n} \\[0.8em]
    &\leq \begin{cases}
    \mu_{1, n}^{1/2} +  B_p n^{-1/2} \mu_{1, n}^{1/2} \leq \mu_{1, n}^{1/2} +  B_p B_\mu \ \mbox{for} \ 1\leq l \leq r \\[0.8em]
    0 + \norm{q_l \Lambda^n} \leq  r B_p B_\mu \ \mbox{for} \  l > r,
    \end{cases} 
\end{align*}
where the second case is obtained as follows: For $l > r$, let $G_n$ be the matrix of orthonormal row eigenvectors of $p_{r+1}(\Gamma_C^n), ..., p_{n}(\Gamma_C^n)$ which are in the kernel of $\Lambda^n$ and let $P_n^C$ the matrix of orthonormal row eigenvectors $p_1(\Gamma_C^n), ..., p_r(\Gamma_C^n)$ and $P_n^y$ analogously. From linear projection on the row space of $G_n$, with residual $v_l$ (dependence on $n$ is omitted in the notation),
\begin{align*}
    p_l &= \proj\left(p_l \mid \text{row}(G_n)\right) + v_l \\
   v_l &= p_l - p_l G_n'G_n = p_l \left(I_n - G_n'G_n\right) = p_l\left[P_n^C\right]' P_n^C \\
    &= p_l\left[P_n^C - P_n^y\right]' P_n^C = p_l \begin{bmatrix}
        q_1' & \cdots q_r'
    \end{bmatrix} P_n^C. 
\end{align*}
Consequently
\begin{align*}
    q_l \Lambda^n &= \left[p_l - p_l^C\right] \Lambda^n = \left[ \proj\left(p_l \mid \text{row}(G_n)\right) + v_l - p_l^C\right] \Lambda^n = v_l \Lambda^n \\[0.8em]
    & \mbox{where} \norm{v_l} = \norm{\begin{pmatrix}
        p_l q_1' & \cdots & p_l q_r'
    \end{pmatrix} P_n^C} = \norm{\sum_{j = 1}^r p_l q_j' p_j^C} \\
    & \leq \sum_{j = 1}^r \norm{p_l q_j'} \norm{p_j^C} \leq \sum_{j = 1}^r \norm{q_j} \leq r  n^{-1/2} B_p\ . 
\end{align*}
This yields 
\begin{align*}
    \norm{q_l \Lambda^n} \leq \norm{v_l} \mu_{1, n}^{1/2} \leq r B_p n^{-1/2} \mu_{1, n}^{1/2} = r B_p B_\mu \quad l > r, n \in \mathbb N.  
\end{align*}
Therefore for all $n\in \mathbb N$, 
\begin{align*}
  (I)_{lj} &\leq \begin{cases}
\left(\mu_{1, n}^{1/2} (1 +  B_p n^{-1/2})\right)^2 \mu_1\left(\hat \Gamma_F - I_r\right)   \ \mbox{for} \ 1 \leq l \leq r \\[0.8em]
\left(\mu_{1, n}^{1/2} (1 + B_p n^{-1/2})\right) r B_p n^{-1/2}  \mu_1\left(\hat \Gamma_F - I_r\right)  \ \mbox{for} \ l > r.
    \end{cases} \\[0.8em]
    n^{-1}(I)_{lj} &\leq \begin{cases}
    \left(B_\mu^2 +n^{-1/2} B_\mu^2 B_p + n^{-1}(B_p B_\mu)^2 \right) \hat B_F   \ \mbox{for} \ 1 \leq l \leq r \\[0.8em]
    n^{-1} r \left(B_p B_\mu + B_p^2 B_\mu \right)\hat B_F  \ \mbox{for} \ l > r.
        \end{cases}
\end{align*}
Next, by A\ref{A: Var est}, we have that 
\begin{align*}
    (II)_{lj} = p_l(\Gamma_y^n) \Lambda^n \hat \Gamma_{Fe} p_j(\Gamma_y^n)' &\leq \begin{cases}
        \left(\mu_{1,n}^{1/2} + \mu_{1, n}^{1/2} B_p n^{-1/2}\right)\hat B_{Fe}  \ \mbox{for} \ 1\leq l \leq r \\[0.8em]
        \mu_{1, n}^{1/2} r B_p n^{-1/2} \hat B_{Fe}\ \mbox{for} \ l > r, 
    \end{cases}\\[0.8em]
    n^{-1}(II)_{lj} &\leq \begin{cases}
       \left(n^{-1/2}B_\mu  + n^{-1}B_p B_\mu\right) \hat B_{Fe} \ \mbox{for} \ 1\leq l \leq r \\[0.8em]
        n^{-1} r B_p B_\mu \hat B_{Fe}\ \mbox{for} \ l > r, 
    \end{cases}
\end{align*}
and 
\begin{align*}
    (III)_{lj} = p_l(\Gamma_y^n) \hat \Gamma_{eF}^n (\Lambda^n)' p_j(\Gamma_y^n)' &\leq \begin{cases}
        \left(\mu_{1,n}^{1/2} + \mu_{1, n}^{1/2} B_p n^{-1/2}\right)\hat B_{Fe}  \ \mbox{for} \ 1\leq l \leq r \\[0.8em]
         \left(\mu_{1,n}^{1/2} + \mu_{1, n}^{1/2} B_p n^{-1/2}\right) \hat B_{Fe}  \ \mbox{for} \ l > r. 
    \end{cases}\\[0.8em]
    n^{-1}(III)_{lj} &\leq \begin{cases}
        \left(n^{-1/2} B_\mu + n^{-1}B_\mu B_p\right)\hat B_{Fe}  \ \mbox{for} \ 1\leq l \leq r \\[0.8em]
         \left(n^{-1/2} B_\mu + n^{-1} B_p B_\mu\right) \hat B_{Fe}  \ \mbox{for} \ l > r. 
    \end{cases}
\end{align*}
Finally $(IV)_{lj} \leq \hat B_e$ from  A\ref{A: Var est}. 

Therefore putting things together with $n^{-1/2} \mu_{1, n}^{1/2} \leq B_\mu$, we obtain
{
\scriptsize
\begin{align*}
    &\sum_{l \neq j, l = 1}^n \left(\frac{t_{lj}}{\mu_j - \mu_l}\right)^2 = \sum_{l \neq j, l = 1}^r \left(\frac{t_{lj}}{\mu_j - \mu_l}\right)^2 + \sum_{l \neq j, l = r+1}^n \left(\frac{t_{lj}}{\mu_j - \mu_l}\right)^2\\[1em]
    &= \sum_{l \neq j, l = 1}^r \left(\frac{n^{-1} \left[(I)_{lj} + (II)_{lj} + (III)_{lj} + (IV)_{lj} \right]}{n^{-1}\left(\mu_j - \mu_l\right)}\right)^2 + \sum_{l \neq j, l = r+1}^n \left(\frac{n^{-1} \left[(I)_{lj} + (II)_{lj} + (III)_{lj} + (IV)_{lj} \right]}{n^{-1}\left(\mu_j - \mu_l\right)}\right)^2 \\[0.8em]
    &\leq r \frac{1}{\inf_n\min_{l\neq j} n^{-1} (\mu_j - \mu_l)}\left[
    \left(B_\mu^2 + n^{-1/2} B_\mu^2 B_p + n^{-1} (B_p B_\mu)^2\right)\hat B_F + 2 (n^{-1/2} B_\mu + n^{-1} B_p B_\mu) \hat B_{Fe}  + n^{-1} \hat B_e
    \right]^2 \\
    &+(r - n) \frac{1}{\inf_n\min_{l\neq j} n^{-1} (\mu_j - \mu_l)} \left(
    n^{-1} r \left(B_p B_\mu + B_p^2 B_\mu \right)\hat B_F + 
    n^{-1} r B_p B_\mu \hat B_{Fe} + \left(n^{-1/2} B_\mu + n^{-1} B_p B_\mu\right) \hat B_{Fe} + n^{-1} \hat B_e
    \right)^2\\
    &= \mathcal O_p(T^{-1/2})^2 + (r-n)\left[\mathcal O(n^{-1/2}) \mathcal O_p (T^{-1/2})\right]^2 = \mathcal O_p(T^{-1}).  
\end{align*}
}
(ii) Again, we employ \cite{wilkinson1965algebraic}, section 2, equation (5.5) together with (\ref{eq: terms for evec approx.}):
{
\footnotesize
\begin{align*}
    &\frac{\mu_j(\hat \Gamma_y^n)}{n} - \frac{\mu_j(\Gamma_y^n)}{n} = k_1 \varepsilon_n + k_2 \varepsilon_n^2 + \cdots \\
\mbox{with}\ &k_1 = p_j(\Gamma_y^n) \left(\hat \Gamma_y^n - \Gamma_y^n \right) p_j' (\Gamma_y^n) = n^{-1}t_{jj} = n^{-1}\left[(I)_{jj} + (II)_{jj}  + (III)_{jj} + (IV)_{jj}\right] \quad (\mbox{for} j \leq r) \\
&= r\left[  \left(B_\mu^2 +n^{-1/2} B_\mu^2 B_p + n^{-1}(B_p B_\mu)^2 \right) \hat B_F  +
\left(n^{-1/2}B_\mu  + n^{-1}B_p B_\mu\right) \hat B_{Fe} +   \left(n^{-1/2} B_\mu + n^{-1}B_\mu B_p\right)\hat B_{Fe} +\hat B_e \right] \\
&= \mathcal O_p(T^{-1/2}).
\end{align*}
}
Furthermore, setting $\hat \mu_j \equiv \hat \mu_j(\hat \Gamma_y^n), \mu_j \equiv \mu_j(\Gamma_y^n)$
\begin{align*}
    \left|\left(\frac{\hat \mu_j}{n}\right)^{-1} - \left(\frac{\mu_j}{n}\right)^{-1}  \right| \leq \left|\left(\frac{\hat \mu_j}{n}\right)^{-1} \right| \left|\frac{\hat \mu_j}{n} - \frac{\mu_j}{n}\right| \left|\left(\frac{\mu_j}{n}\right)^{-1}\right| = \mathcal O_p(1) \mathcal O_p(T^{-1/2}) \mathcal O(1) = \mathcal O_p(T^{-1/2}).
\end{align*}
Finally, since for all $a,b \in \mathbb R$ with $a+b \neq 0$, we have $(a^2 - b^2)/(a+b) = (a-b)$, we conclude
\begin{align*}
|\hat \mu_j^{-1/2} - \mu_j^{-1/2}| &= n^{-1/2} \left | \left(\frac{\hat \mu_j}{n}\right)^{-1/2}- \left(\frac{\mu_j}{n}\right)^{-1/2}\right| \\
&=n^{-1/2}  \left|\left(\frac{\hat \mu_j}{n}\right)^{-1} - \left(\frac{\mu_j}{n}\right)^{-1}  \right| \left | \left(\frac{\hat \mu_j}{n}\right)^{1/2} +\left(\frac{\mu_j}{n}\right)^{1/2}\right|^{-1} = n^{-1/2} \mathcal O_p(T^{-1/2}). 
\end{align*}
(iii) For the last statement observe that 
\begin{align*}
        \norm{\hat{\mathcal K}_j^{F, n} - \mathcal K_j^{F, n}} &= \norm{\hat \mu_j^{-1/2} \hat p_j^n - \mu_j^{-1/2} p_j^n} \\
        &= \norm{\hat p_j^n - p_j} |\mu_j^{-1/2}| + |\hat \mu_j^{-1/2} - \mu_j^{-1/2}| \norm{p_j^n} + \norm{\hat p_j^n - p_j^n}|\hat \mu_j^{-1/2} - \mu_j^{-1/2}| \\
        &= \mathcal O_p(T^{-1/2})\mathcal O(n^{-1/2}) + \mathcal O_p(T^{-1/2} n^{-1/2}) +
         \mathcal O_p(T^{-1}n^{-1/2}) = \mathcal O_p(T^{-1/2}n^{-1/2}).
  \end{align*}  
\end{proof}
\begin{theorem}[Consistency of Factors, Factor space and Common Component]
Under assumptions A\ref{A: r-SFS struct} to A\ref{A: Var est}, we have
    \begin{itemize}
        \item[(i)] $\norm{\hat F_t^n - F_t^{\infty}} = \mathcal O_{ms}(\max(T^{-1/2}, n^{-1/2}))$
        \item[(ii)] $\norm{\hat F_t^n - \hat H_n F_t} = \mathcal O_{ms}(n^{-1/2})$ for $\hat H_n = T^{-1}\sum_{t = 1}^T \hat F_t^n  F_t' \left(T^{-1} \sum_{t = 1}^T F_t F_t'\right)^{-1}$ with finite $T < \infty$.
        %
        %
    \end{itemize}
\end{theorem}
\begin{proof}
    \begin{align*}
        \norm{\hat F_t^n - F_t^{\infty}} \leq \norm{\hat F_t^n - F_t^n} + \norm{F_t^n - F_t^{\infty}} = \mathcal O_{ms}(T^{-1/2}) + \mathcal O_{ms}(n^{-1/2})
    \end{align*}
    where for the first term we use Cauchy-Schwarz inequality
    \begin{align*}
       \norm{\hat F_{tj}^n - F_{j}^n} &= \norm{(\hat{\mathcal K}_j^{F, n}  - \mathcal K_j^{F,n}) y_t^n} \leq  \norm{\hat{\mathcal K}_j^{F, n}  - \mathcal K_j^{F,n}}\norm{y_t^n} \\
       &= \mathcal O_{ms}(T^{-1/2}n^{-1/2}) \mathcal O_{ms}(n^{1/2}) = \mathcal O_{ms}(T^{-1/2}).
    \end{align*}
    For statement (ii), note that 
    \begin{align*}
        \frac{1}{T} \sum_{t = 1}^T \hat F_t^n F_t' = \frac{1}{T} \sum_{t = 1}^T \hat{\mathcal K}^{F, n} \left(\Lambda^n F_t + e_t^n \right) F_t' = \hat{\mathcal K}^{F, n} \Lambda^n \hat \Gamma_F + \hat{\mathcal K}^{F, n} \hat \Gamma_F
    \end{align*}
    Consequently, 
    \begin{align*}
        \hat F_t^n - \hat H_n F_t &= \hat{\mathcal K}^{F, n} \Lambda^n F_t + \hat{\mathcal K}^{F, n} e_t^n - \hat{\mathcal K}^{F, n} \Lambda^n \hat \Gamma_F \hat \Gamma_F^{-1} F_t - \hat{\mathcal K}^{F, n} \hat \Gamma_{eF} \hat \Gamma_F^{-1} F_t \\
        \norm{\hat F_t^n - \hat H_n F_t} &\leq \norm{\hat{\mathcal K}^{F, n} e_t^n} + \norm{\hat{\mathcal K}^{F, n} \hat \Gamma_{eF}}\norm{\hat \Gamma_F^{-1}}\norm{F_t} = \mathcal O_{ms}(n^{-1/2}). 
    \end{align*}
\end{proof}
\section{Conclusion}
We refine the interpretation of what the principal component estimator is targeting by showing that the population normalised principal components actually have a limit in mean square. This is tied to the assumption that the eigenvalues of the common component variance matrix divided by $n$ are distinct which allows convergence of the eigenvectors. 
\bibliographystyle{apalike} 
\bibliography{references.bib}
\appendix
\renewcommand{\theequation}{\thesection.\arabic{equation}}

\end{document}